\documentclass[prd,tightenlines,nofootinbib,superscriptaddress]{revtex4}

\usepackage{amsfonts,amssymb,amsthm,bbm}

\usepackage{amsmath}

\usepackage{hyperref}

\usepackage{color,psfrag}
\usepackage[dvips]{graphicx}
\usepackage{tikz}
\usetikzlibrary{calc}
\usetikzlibrary{decorations.pathmorphing}

\newcommand{\C}{{\mathbb C}}
\newcommand{\N}{{\mathbb N}}
\newcommand{\R}{{\mathbb R}}

\newcommand{\cA}{{\mathcal A}}

\newcommand{\cJ}{{\mathcal J}}

\newcommand{\cN}{{\mathcal N}}

\newcommand{\cO}{{\mathcal O}}

\newcommand{\cD}{{\mathcal D}}
\newcommand{\cC}{{\mathcal C}}
\newcommand{\cS}{{\mathcal S}}

\newcommand{\SU}{\mathrm{SU}}
\newcommand{\Sp}{\mathrm{Sp}}

\newcommand{\SL}{\mathrm{SL}}
\newcommand{\SO}{\mathrm{SO}}

\newcommand{\U}{\mathrm{U}}

\newcommand{\be}{\begin{equation}}
\newcommand{\ee}{\end{equation}}
\newcommand{\beq}{\begin{eqnarray}}
\newcommand{\eeq}{\end{eqnarray}}
\newcommand{\bes}{\begin{eqnarray}}
\newcommand{\ees}{\end{eqnarray}}

\newcommand{\mat} [2] {\left ( \begin{array}{#1}#2\end{array} \right ) }

\newcommand{\su}{{\mathfrak{su}}}
\renewcommand{\sp}{{\mathfrak{sp}}}

\renewcommand{\sl}{{\mathfrak{sl}}}
\newcommand{\so}{{\mathfrak{so}}}

\newcommand{\la}{\langle}
\newcommand{\ra}{\rangle}

\newcommand{\f}{\frac}

\def\nn{\nonumber}
\def\pp{\partial}
\newcommand{\w}{\wedge}

\def\vphi{\varphi}
\def\eps{\epsilon}

\newcommand{\id}{\mathbb{I}}

\def\act{\triangleright}

\def\ve{\vec{e}}
\def\vx{\vec{x}}
\def\vu{\vec{u}}
\def\vp{\vec{p}}
\def\vv{\vec{v}}
\def\vC{\vec{C}}
\def\vcJ{\vec{\cJ}}
\def\vJ{\vec{J}}
\def\vK{\vec{K}}
\def\vL{\vec{L}}
\def\vN{\vec{N}}
\def\vsigma{\vec{\sigma}}
\def\arr{\rightarrow}
\def\act{\,\triangleright\,}
\def\bz{\bar{z}}

\def\he{\hat{e}}
\def\hv{\hat{v}}
\def\hp{\hat{p}}
\def\hJ{\hat{J}}

\def\dd{\mathrm{d}}

\newtheorem{theorem}{Theorem}[section]
\newtheorem{lemma}[theorem]{Lemma}
\newtheorem{prop}[theorem]{Proposition}

\def\centerarc[#1](#2)(#3:#4:#5)
    { \draw[#1] ($(#2)+({#5*cos(#3)},{#5*sin(#3)})$) arc (#3:#4:#5); }
    
\def\centerarcnodes[#1](#2)(#3:#4:#5)(#6,#7)
    {\coordinate(#6) at ($(#2)+({#5*cos(#3)},{#5*sin(#3)})$);
    \coordinate(#7) at ($(#2)+({#5*cos(#4)},{#5*sin(#4)})$);
    \draw[#1] ($(#2)+({#5*cos(#3)},{#5*sin(#3)})$) arc (#3:#4:#5); }

\def\angcircle(#1)(#2)(#3:#4) {\coordinate(#1) at ($(#2)+({#4*cos(#3)},{#4*sin(#3)})$); }

\begin{document}

\title{Bubble Networks: Framed Discrete Geometry for Quantum Gravity}

\author{{\bf Laurent Freidel}}\email{lfreidel@perimeterinstitute.ca}
\affiliation{Perimeter Institute for Theoretical Physics, 31 Caroline Street North, Waterloo, Ontario, Canada N2L 2Y5}

\author{{\bf Etera R. Livine}}\email{etera.livine@ens-lyon.fr}
\affiliation{Perimeter Institute for Theoretical Physics, 31 Caroline Street North, Waterloo, Ontario, Canada N2L 2Y5}
\affiliation{Université de Lyon, ENS de Lyon, UCBL, CNRS,
Laboratoire de Physique, 69342 Lyon, France}

\date{\today}

\begin{abstract}

In the context of canonical quantum gravity in 3+1 dimensions, we introduce a new notion of bubble network that represents discrete 3d space geometries. These are natural extensions of twisted geometries,
which represent  the geometrical data underlying loop quantum geometry and are defined as networks of $\SU(2)$ holonomies.
In addition to the $\SU(2)$ representations encoding the geometrical flux, the bubble network links carry a compatible $\SL(2,\R)$ representation encoding the discretized frame field which composes the flux. In contrast with twisted geometries, this extra data allows to reconstruct the frame compatible with the flux unambiguously.
At the classical level this data represents a network of 3d geometrical cells glued together. The $\SL(2,\R)$ data contains information about the discretized 2d metrics of the interfaces between 3d cells and $\SL(2,\R)$ local transformations are understood as the group of area-preserving diffeomorphisms.
We further show that the natural gluing condition with respect to this extended group structure ensures that the intrinsic 2d geometry of a boundary surface is the same from the viewpoint of the two cells sharing it.
At the quantum level this gluing corresponds to a  maximal entanglement along the network edges.
We  emphasize that the nature of this extension of twisted geometries is compatible with the general analysis of gauge theories that predicts edge mode degrees of freedom at the interface of subsystems.

\end{abstract}

\maketitle
\tableofcontents

\section*{Introduction}

The goal of quantum gravity is to create a mathematically-consistent description of the space-time geometry unifying the visions of general relativity and quantum theory. Here we work in a canonical framework relying on a decomposition of the four-dimensional space-time as a 3d space geometry evolving in time. The  operational perspective we are developping is to decompose the 3d geometry into 3d cells, similarly to the decomposition of a manifold into charts. This consists in defining the state of geometry for each chunk of 3d geometry and describing the consistency conditions necessary to glue those 3d cells together in order to form the overall 3d geometry.

The central challenge of such a procedure is to understand what could be the geometrical elements one should keep at the most fundamental discrete level in order to capture the key symmetry of the theory, that is diffeomorphisms. While the full answer to that question is still awaiting, there has been important recent progress in that direction.
The key idea is to first relate the process of discretizing a gravitational system as being dual to the process of subdividing a continuum gravitational system into simpler elements \cite{Freidel:2011ue,Freidel:2013bfa}.
The second point is the understanding that, when one subdivides a gauge theory, the gauge symmetries are then promoted to local boundary symmetries \cite{Freidel:2015gpa,Donnelly:2016auv,Rovelli:2013fga} . The mechanism behind this is that the presence of gauge symmetry in the total system reveals boundary degrees of freedom - {\it edge modes} - along the subdivision cut and these edge modes  form a representation of the boundary symmetry group \cite{Donnelly:2016auv}.
At the discrete level, this boundary symmetry group is attached to each 2d interface between 3d cells (or equivalently the dual link) as a remnant of the continuous gauge symmetry. 
The question is then to understand what is  the proper boundary symmetry group for  gravitational edge modes.

In the context of the first order formulation of general relativity, as used e.g. in loop quantum gravity, it has been clear for  a while that the boundary symmetry group should include the local $\SU(2)$ group descending from the local  Lorentz gauge  transformations.
More recently, it has further been  understood that the boundary symmetry  should also include the group of area-preserving diffeomorphisms of the boundary, which is isomorphic to a local $\SL(2,\R)$ group.
More precisely, it has been shown in \cite{Freidel:2015gpa} that the generators of this $\SL(2,\R)$ group are given by the components of the 2-dimensional metric on each 2d interface between 3d cells. From this perspective, the extension from $\SU(2)$ to $\SU(2)\times\SL(2,\R)$ appears necessary  in order to include diffeomorphisms at the fundamental level.
In this work we show how to implement this idea at the level of discrete 3d geometries and how it can naturally be viewed as an extension of the notion of twisted geometry  \cite{Freidel:2010aq} and its spinor implementation \cite{Girelli:2017dbk,Freidel:2010tt,Livine:2011gp,Livine:2011zz}.  Conversely we show that one recovers  twisted geometries when one fixes the conformal gauge for the 2d metrics of the two dimensional interfaces.

Refining our description of the 3d geometry, we advocate to consider every 3d cells as bubbles, meaning that we will describe the 3d geometry of each cells as the state of the 2d geometry of its boundary. Then the 3d cells are glued along shared boundary surfaces and consistency conditions turn into matching constraints between the two descriptions of the geometry of the boundary from the perspective of the two 3d cells sharing it. This picture leads to 3d geometry as a networks of bubbles, see fig.\label{fig:soapbubble}, similar in spirit to the cellular decompositions used to formulate discrete topological quantum field theories and topological state-sums.
  \begin{figure}[h!]
	\centering
		\includegraphics[scale = 0.25]{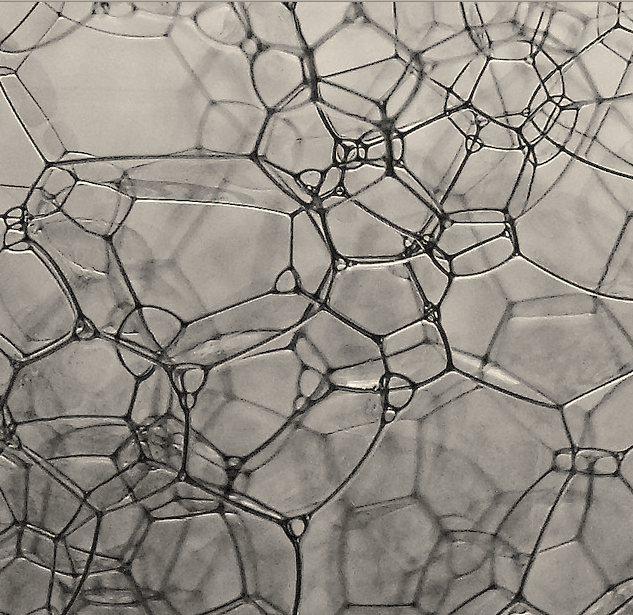}	
	\caption{Example of a bubble network, experimentally realized with soap bubbles: bubbles fill up the 3d space and are glued along surface patches. \label{fig:soapbubble}}
\end{figure}

On the one hand, such a description fits with the idea of a quasi-local holography in quantum gravity: the dynamically-relevant degrees of freedom of the 3d bulk geometry within a cell would be entirely encoded in the state of 2d geometry on the cell's boundary. On the other hand, it can also be interpreted from a coarse-graining point of view. The 3d space is thought as made of elementary chunks of a fixed given 3d geometry, e.g. flat or homogeneously curved 3d geometry, so that the only remaining degree of freedom is the embedding data of the 2d boundary surface within this 3d geometry. Then as one would merge those bubbles together to define a coarse-graining flow towards a  description of this 3d geometry at larger scale, one would define larger bubbles from gluing smaller bubbles together, coarse-grain the bulk geometry of those bubbles and derive the 2d geometry state of the larger bubble surfaces. For instance, this is the standard interpretation of spin network states for 3d geometry defined in Loop Quantum Gravity and interpreted as the quantum counterpart of  discrete twisted geometries \cite{Freidel:2010aq}. Finally, once we have a consistent definition of bubble networks, the goal would be to formulate a quantum version of relativity, i.e. understand how changes of observers, through diffeomorphisms and scale transformations, translate to changes in the bubble network, as modifying the atlas of the 3d geometry and coarse-graining or refining the bubbles.

\medskip

In the present paper, we present the definition of bubble networks as extended twisted geometries. Twisted geometries were introduced as a generalization of Regge triangulations to describe the discrete classical geometry of spin network states in loop quantum gravity \cite{Freidel:2010aq,Freidel:2010bw}. They describe the geometry of a graph dressed with $\SU(2)$ group elements on its edges. The data associated with each graph node can be interpreted in terms of a chunk of 3d volume whose geometry is encoded into the area flux of its surface elements. In this framework  an edge linking two nodes is dual to the boundary surface between the two corresponding 3d cells, and the $\SU(2)$ group element living on that edge gives the mapping between the associated two 3d reference frames.
In the twisted geometry interpretation, part of the $\SU(2)$ group element ensures the matching of the normal vector to the boundary surface between two neighboring 3d cells, while the remaining degree of freedom, called ``twist angle'', is understood to encode the extrinsic curvature \cite{Dittrich:2008ar,Dittrich:2008va,Dupuis:2012yw}. This allows to represent the phase space of loop quantum gravity where the Ashtekar-Barbero connection, used as configuration variable, actually mixes the intrinsic geometry (through the 3d spin-connection) and the extrinsic curvature. In this context, the identification of the twist angle has been crucial in understanding the embedding of the canonical discrete 3 geometry into the 4d space-time geometry and was translated into an embedding of the $\SU(2)$ group structure into $\SL(2,\C)$ specially useful to clarify the reformulation of the dynamics of loop quantum gravity in terms of spinfoam path integrals \cite{Speziale:2012nu}.
%
%
Besides these achievements, one major drawback of the twisted geometry picture is that although geometrical flux agrees, it generally gives inconsistent frame geometries \cite{Haggard:2012pm} across the boundaries\footnote{We can think of fluxes as discrete analog  Lie algebra valued 2-forms, while frames are the discrete analog of  Lie algebra valued 1-form.}

In order to correct this inconvenient feature of twisted geometry,  we propose to extend the $\SU(2)$ structure to a $\SU(2)\times\SL(2,\R)$ structure, where the new $\SL(2,\R)$ sector will encode the 2d geometry state of each bubble, i-e the discrete frame field,  while the $\SU(2)$ sector still encodes the area geometry and its transport from one bubble to the neighboring bubbles. The matching constraints resulting the gluing of bubbles which expresses that the area derived form the frame agrees with the area derived from the flux, translate into a Casimir balance equation, between the $\SU(2)$ sector and the $\SL(2,\R)$ sector. This relates natural  algebraic condition on Casimirs with  the compatibility condition between the intrinsic and extrinsic geometries of each bubble.

From the point of view of loop quantum gravity, this consists in an extension of the classical phase space in order to take seriously the dual surface interpretation of intertwiners for spin network states. Indeed, the nodes of a spin network states carry intertwiners, i.e. $\SU(2)$-invariant tensors, which are interpreted as quantum (convex) polyhedra (embedded in flat 3d space) \cite{Bianchi:2010gc,Livine:2013tsa,Barbieri:1997ks,Baez:1999tk}. Nevertheless, when coarse-graining loop quantum gravity, curvature naturally builds up at the spin network nodes \cite{Livine:2013gna,Dittrich:2014wpa,Dittrich:2014wda} and it becomes necessary to allow for ``curved nodes'', corresponding to quantum curved polyhedra (e.g. embedded in spherical or hyperbolic space) \cite{Charles:2016xwc,Charles:2016xzi}, and to allow for the 3d embedding of the boundary surface to fluctuate, thereby allowing for bubbles with arbitrary geometry.
The line of research we pursue in the present work parallels the logic developed in \cite{Freidel:2015gpa,Freidel:2016bxd}, which couples surface geometry degrees of freedom to the Ashtekar triad-connection variables leading to  coupled spin networks and conformal field theories at the quantum level. Nevertheless, we start with discrete surface geometries for the bubbles, so that the framework we propose is much simpler, though less rich, than the ``loop gravity string '' picture introduced in \cite{Freidel:2016bxd}.
It nevertheless admits a clear geometrical interpretation and a clear mapping back to the usual twisted geometry framework through a straightforward gauge-fixing procedure.

In the first section, we  start with a quick review of twisted geometries. Then based on the canonical analysis of the symplectic structure for the 2d geometry on boundary surfaces worked out in \cite{Freidel:2015gpa}, we will introduce a discretized version of the 2d geometry for the bubbles and define the phase space for the bubble networks. In the second section, we will show how a symplectic reduction of the bubble network phase space leads back to twisted geometries. We will further interpret this as the choice of the conformal gauge for the 2d geometry on the bubble surfaces.

\section{Discrete Bubble Networks}

The goal of this section is to introduce the discretized geometry of bubble networks and provide a clean mathematical definition of gluing bubbles. As a reference, we start by recalling the algebraic structure of twisted geometries in loop quantum gravity and their phase space. We then turn to bubble networks, describe the phase space of the 2d geometry on the bubbles' surfaces and construct bubble networks. They will turn out to be realized as twisted geometries augmented with the intrinsic geometry data of the bubbles' surface. This provides the discrete network version of the metric-flux algebra introduced in \cite{Freidel:2015gpa} as an upgrade for the holonomy-flux algebra underlying loop quantum gravity.

\subsection{A quick review of twisted geometry and spin(or) networks}

Twisted geometries are the discrete geometrical structure underlying loop quantum gravity.
They are networks of $\SU(2)$ holonomies encoding changes of reference frames and parallel transport from one point of space to another.  They form the classical structure for spin network states of quantum geometry.
More precisely, they are defined with reference to a graph, or network, dressed with algebraic data. Let us consider an oriented (closed) graph $\Gamma$. The twisted geometry phase space on the graph $\Gamma$ is defined as one copy of the $T^*\SU(2)$ phase space attached to each link of the graph, together with closure constraints  generating a $\SU(2)$ gauge invariance at each node. For each edge $e$, we write $s(e)$ and $t(e)$ respectively for the source and target nodes of the edge. Then we introduce one $\SU(2)$ group element $h_{e}\in\SU(2)$ for each edge $e$ and a pair of 3-vectors $\vJ_{e}^s$ and $\vJ_{e}^t$ on each edge corresponding to respectively to its source and target vertices. The $\SU(2)$ group elements are considered as 2$\times$2 matrices. From the viewpoint of a vertex $v$, we therefore have one vector $\vJ_{e}^v$ for each edge $e$ attached to $v$. We endow those variables with the $T^*\SU(2)$ Poisson brackets:
\be
\label{hJalgebra}
\left|\begin{array}{lcl}
\{(J^{s}_{e})^{a},(J^{s}_{e})^{b}\}&=& \eps_{abc}(J^{s}_{e})^{c}
\,,
\vspace*{1mm}\\
\{(J^{t}_{e})^{a},(J^{t}_{e})^{b}\}&=& -\eps_{abc}(J^{t}_{e})^{c}
\,,
\end{array}\right.
\qquad
\left|\begin{array}{lcl}
\{\vJ^{s}_{e},h_{e}\}&=&\f i2 h_{e}\vsigma
\,,
\vspace*{1mm}\\
\{\vJ^{t}_{e},h_{e}\}&=& \f i2\vsigma h_{e}
\,,
\end{array}\right.
\qquad
\left|\begin{array}{lcl}
\{\vJ^{s}_{e},\vJ^{t}_{e}\}&=&0
\,,
\vspace*{1mm}\\
\{h_{e},h_{e}\}&=& 0
\,,
\end{array}\right.
\ee
where $\eps_{abc}$ is the rank-3 totally antisymmetric tensor and the $\sigma^{a}$'s are the Pauli matrices, normalized such that they square to the identity, $(\sigma_{a})^2=\id_{2}$, and with commutation relations $[\sigma_{a},\sigma_{b}]=2i\eps_{abc}\sigma_{c}$.
We supplement these brackets with two sets of constraints:
\begin{itemize}
\item a matching constraint along each edge $e$, imposing that each target vector is the transport of the source vector by the group element along the edge:
\be
\vJ^t_{e}=h_{e}\triangleright\vJ^s_{e}\,,
\quad \mathrm{or}\quad
\vJ^t_{e}\cdot\vsigma=h_{e}\,(\vJ^s_{e}\cdot\vsigma)\,h_{e}^{-1}
\ee

\item a closure constraint around each vertex $v$, imposing that the oriented sum of all the vectors vanishes:
\be
\vC_{v}=\sum_{e\ni v}\eps^e_{v}\vJ_{e}^v=0\,,
\ee
where the sign $\eps^{e}_{v}=\pm$ registers the relative orientation of the edge with respect to the vertex, positive if the edge is outgoing  $v=s(e)$ and negative if the edge is incoming $v=t(e)$.

\end{itemize}
The matching constraints effectively reduce the variables attached to each edge to the $T^*\SU(2)$ phase space, while the closure constraint at a vertex $v$ generates $\SU(2)$ gauge transformations $\exp\{\vu\cdot\vC_{v},\bullet\}$ acting on all the vectors around that vertex:
\be
\begin{array}{l}
\forall e \textrm{ such that } v=s(e)
\,,
\quad
\vJ_{e}^v
\,\longmapsto\,
k_{v}\triangleright\vJ_{e}^v
\,,
\quad
h_{e}
\,\longmapsto\,
h_{e}k_{v}^{-1}
\vspace*{2mm}\\
\forall e \textrm{ such that } v=t(e)
\,,
\quad
\vJ_{e}^v
\,\longmapsto\,
k_{v}\triangleright\vJ_{e}^v
\,,
\quad
h_{e}
\,\longmapsto\,
k_{v}h_{e}
\end{array}
\ee
The standard geometrical interpretation of the closure constraint is that $N$ vectors $\vJ_{1},..,\vJ_{N}\in\R^3$ whose sum vanishes uniquely determine a convex polyhedra in the flat space $\R^3$ with $N$ faces such that the $\vJ_{i}$ are the normal vectors to each face of the polyhedron with their norm giving the area of the corresponding face (see e.g. \cite{Bianchi:2010gc}). The twisted geometry is then the collection of such polyhedra glued together by the matching constraints. The gluing of two neighboring polyhedra is loose in the sense that it only requires a matching of the area of the glued faces but not an actual matching of their precise shape.

The spinning geometry interpretation of this structure is more flexible and sidesteps the shape mismatch issue \cite{Freidel:2013bfa}.
The 3d geometry is constructed as a cellular complex from 3d flat cells whose boundary 2-cells are minimal surfaces between 1-cells. The ``normal vectors'' $\vJ$'s of the surfaces are defined as an integrated angular momentum, computed as the holonomy of a specific connection around the surfaces (see also \cite{Charles:2016xzi}). This holonomy is matched across the boundary when gluing two bubbles and is not generically the normal vector to a flat 2d face.

\medskip

A useful reparametrization of the twisted geometry phase space is in terms of spinor networks \cite{Freidel:2010bw,Borja:2010rc,Livine:2011gp,Livine:2011zz,Dupuis:2012vp}.
This spinorial parametrization of  twisted geometries led to a systematic construction of coherent intertwiners \cite{Freidel:2010tt,Girelli:2017dbk} and semi-classical spin network states \cite{Dupuis:2011fz,Bonzom:2012bn} and remarkable exact computations of spinfoam transition amplitudes \cite{Freidel:2012ji,Bonzom:2015ova}.
One introduces a complex 2-vector, or spinor, $z_{e}^v\in\C^2$ on each half-edge. Each spinor is endowed with a canonical Poisson bracket:
\be
\{z_{A},\bz_{B}\}=-i\delta_{AB}\,,
\ee
while  spinor components living on different half-edges commute with each other. We use the standard ket notation for the spinors, writing $|z\ra\in\C^2$ with the corresponding  dual spinor denoted $|z]$:
\be
|z\ra=\mat{c}{z_{0}\\z_{1}}
\,,\qquad
\la z |=\mat{cc}{\bz_{0}&\bz_{1}}
\,,\qquad
|z]=\mat{cc}{0 & -1 \\ 1 & 0}|\bz\ra=\mat{c}{-\bz_{1}\\\bz_{0}}
\,,\qquad
[z|=\mat{cc}{-z_{1}&z_{0}}
\,.
\ee
We can define both the vectors and the $\SU(2)$ group elements in terms of the spinors:
\be
\vJ_{e}^s=\f12\la z_{e}^s|\vsigma|z_{e}^s\ra
\,,\quad
\vJ_{e}^t=\f12[ z_{e}^t|\vsigma|z_{e}^t]=-\f12\la z_{e}^t|\vsigma|z_{e}^t\ra
\,,\qquad
h_{e}=\f{|z^t_{e}]\la z_{e}^s|-|z^t_{e}\ra[ z_{e}^s|}{\sqrt{\la z_{e}^s|z_{e}^s\ra\la z_{e}^t|z_{e}^t\ra}}
\,.
\ee
Upon assuming a norm-matching condition along every edge, $\la z_{e}^s|z_{e}^s\ra=\la z_{e}^t|z_{e}^t\ra$, this definition ensures that the $h_{e}$'s lay in $\SU(2)$, i.e. $h_{e}^\dagger=h_{e}^{-1}$,  and that these group elements maps the source spinors onto the dual of the target spinors, $h_{e}\,|z_{e}^s\ra=\,|z_{e}^t]$, thus  mapping the source vector onto the target vectors, $h_{e}\triangleright \vJ_{e}^s=\vJ_{e}^t$. Moreover, one can check that these definitions imply that the $T^*\SU(2)$ Poisson brackets \eqref{hJalgebra} between the $h_{e}$ and the vectors $\vJ_{e}^v$ are weakly satisfied assuming the norm-matching conditions for spinors \cite{Freidel:2010bw}. 
This means that the symplectic quotient of  the spinorial phase space $(\C^4)^{\times E}$ (where $E$ is the number of edges of the graph $\Gamma$)  by the norm-matching conditions gives $(T^*\SU(2))^{\times E}$. Then we still impose the closure constraints and quotient by the resulting $\SU(2)$ gauge invariance at every node of the graph.
In this sense, the spinors provide Darboux coordinates for the twisted geometry phase space.

Finally, one can quantize these spinor networks and define wave-functions as holomorphic polynomials in the spinors satisfying the $\SU(2)$-gauge invariance at every node. This leads to the spin networks of loop quantum gravity, with $\SU(2)$ representations along the graph edges and $\SU(2)$ intertwiners at the nodes \cite{Borja:2010rc,Livine:2011gp,Bonzom:2012bn}.

\subsection{Discretization of the surface geometry: $\SU(2)\times \SL(2,\R)$ and the Casimir balance equation}

Now we will start with the phase space for surface degrees of freedom in general relativity, introduce a discretized version of this phase space of 2d geometries on the bubbles' surface and glue them consistently to define the phase space for bubble networks, which will turn out to be the twisted geometry phase space augmented with an additional $\sl_{2}$ structure encoding the 2d  intrinsic geometry of the bubbles boundaries.

As worked out in \cite{Freidel:2015gpa}, the symplectic structure of general relativity in its first order formulation in terms of the Ashtekar-Barbero variables induces a boundary symplectic structure on a 2d boundary surface $\cS$ (on the space-like canonical slice) such that the two components of the triad tangent to the surface are conjugate to one another:
\be
\{e_{1}^a(x),e_{2}^b(y)\}=\delta^{ab}\delta^{(2)}(x,y)
\,,\quad x,y\in \cS
\,.
\ee
One of the main point shown in \cite{Freidel:2015gpa} is that this data allow the reconstruction of an $SU(2)$ flux field by taking the wedge product of frames and a local metric on the sphere
 by considering the scalar products of those two vectors.
\be
X_a(x)\equiv \epsilon_{abc} (e^b_1 e^c_2- e^b_2 e^c_1)(x),\qquad 
q_{AB}(x)\equiv \sum_{a=1}^3e_{A}^ae_{B}^b(x). 
\ee
Moreover the flux field generate an $SU(2)$ algebra and the metric generates an $SL(2,\R)$ algebra. 
The goal is to understand how one can naturally discretize this structure and embed it into the loop gravity phase space.
It will be convenient to express the  2d metric in terms of a matrices of scalar product:
\be
{}^{2d}q
=
\mat{cc}{|\ve_{1}|^2 & \ve_{1}\cdot\ve_{2}\\
\ve_{1}\cdot\ve_{2} &|\ve_{2}|^2}
\,.
\ee

We apply this to each bubble, considered as a piecewise-linear surface made of several flat 2d patches. Then bubbles will be glued together through those patches. Each 2d patch thus carries a pair of vectors in $\R^3$, given by the surface integral of  the triad projected onto the surface. Let us look at a bubble with $N$ patches, and thus with $N$ pairs of vectors $((\ve_{1})_{i},(\ve_{2})_{i})_{i=1..N}$. Since the $\ve_{1}$'s are canonically conjugate to the $\ve_{2}$'s, we find it convenient to write them as $\vx\equiv\ve_{1}$ and $\vp\equiv\ve_{2}$, explicitely reflecting the phase space structure on the bubble:
\be
\{x_{i}^a,p_{j}^b\}=\delta_{ij}\delta^{ab}\,.
\ee
Let us focus on a single surface patch and drop its $i$ label. This is exactly the phase space of a three-dimensional particle, parametrized by the two vectors $\vx$ and $\vp$.
We introduce the angular momentum observables:
\be
J^{ab}=x^{a}p^{b}-x^{b}p^{a}\,,
\quad
J_{c}=\eps^{abc}x^{a}p^{b}\,,
\quad
\{J_{a},J_{b}\}=\eps^{abc}J_{c}\,.
\ee
These are  the generators of the Lie algebra  $\su_{2}\sim\so_{3}$ of 3d rotations, which act as usual as $3\times 3$ matrices on 3d vectors:
\be
\left|
\begin{array}{l}
\{J_{a},x_{b}\}=\eps^{abc}x^{c}\,,
\vspace*{1mm}\\
\{J_{a},p_{b}\}=\eps^{abc}p^{c}\,,
\end{array}
\right.
\qquad
\left|
\begin{array}{l}
e^{\{\vu\cdot\vJ\,,\,\cdot\,\}}\,\vx=\cO\,\vx
\,,
\vspace*{1mm}\\
e^{\{\vu\cdot\vJ\,,\,\cdot\,\}}\,\vp=\cO\,\vp
\,,
\end{array}
\right.
\quad
\cO\,\in\SO(3)\,.
\ee
The Casimir of the $\su_{2}$ algebra is the norm of the angular momentum:
\be
\cC=\vJ^{2}\,,\qquad
\{\cC,J^{a}\}=0\,.
\ee
We further introduce  rotation-invariant observables, given by the norms of $\vx$ and $\vp$ and their scalar product:
\be
\ell_{0}=\vx\cdot\vp\,,\quad
\ell_{-}=\vx^{2}\,,\quad
\ell_{+}=\vp^{2}\,,
\qquad
\{J^{a},\ell_{\beta}\}=0\,.
\ee
We label the $\ell$'s with greek indices to emphasize the difference with the vector indices.The observable $\ell_{0}$ is the generator of  dilatations on the phase space, $(\vx,\vp)\,\arr\,(e^{+\lambda}\vx,e^{-\lambda}\vp)$.
These scalar product observables turn out to form a $\sp_{2}$ algebra:
\be
\{\ell_{0},\ell_{\pm}\}=\pm2\ell_{\pm}\,,\quad
\{\ell_{+},\ell_{-}\}=4\ell_{0}\,.
\ee
This algebra is also isomorphic to the $\sl(2,\R)$ Lie algebra, explicitly realized through a simple change of basis:
\be
j_{3}=\f12(\ell_{-}+\ell_{+})
\,,\quad
k_{1}=\ell_{0}
\,,\quad
k_{2}=\f12(\ell_{-}-\ell_{+})
\,,\quad
\left|
\begin{array}{l}
\{j_{3},k_{1}\}=2k_{2}\,,\\
\{j_{3},k_{2}\}=-2k_{1}\,,\\
\{k_{1},k_{2}\}=-2j_{3}\,.
\end{array}
\right.
\ee
These observables generate linear canonical transformations on the $(\vx,\vp)$ phase  space\footnotemark:
\footnotetext{
The explicit exponentiated action of the $\ell$'s is easily computed as:
$$
e^{\{\lambda_{0}\ell_{0}+\lambda_{+}\ell_{-}+\lambda_{-}\ell_{+},\,\cdot\,\}}
\,\mat{c}{\vx\\ \vp}
=
\mat{cc}{-\lambda_{0} & -\lambda_{-}\\ \lambda_{+} & \lambda_{0}}\,\mat{c}{\vx \\ \vp}
=
M\,\mat{c}{\vx\\ \vp}\,,
\quad
tr M=0
\,,
\quad
M^{2}=(\lambda_{0}^{2} -\lambda_{+}\lambda_{-})\,\id=\Delta\,\id\,,
$$
$$
\Omega=e^{M}=\cosh\sqrt{\Delta}\,\id+\f{\sinh\sqrt{\Delta}}{\sqrt{\Delta}}\,M
\quad\textrm{if}\,\, \Delta>0
\quad\textrm{or}\quad
\cos\sqrt{-\Delta}\,\id+\f{\sin\sqrt{-\Delta}}{\sqrt{-\Delta}}\,M
\quad\textrm{if}\,\, \Delta<0
\,.
$$
}
\be
\mat{c}{\vx\\ \vp}\,\arr\,
\Omega\,\mat{c}{\vx\\ \vp}
\,,\quad
\det_{2\times 2}\Omega=1
\,,\quad
\Omega\in\Sp(2)=\SL(2,\R)\,.
\ee
It is convenient to repackage the $\ell$'s in a $2\times 2$ matrix,
\be
D
=
\mat{c}{\vx\\ \vp}\mat{cc}{\vx& \vp}
=
\mat{cc}{x^{2} & \vx\cdot\vp\\ \vx\cdot\vp & \vp^{2}}
=
\mat{cc}{\ell_{-} & \ell_{0}\\ \ell_{0} & \ell_{+}}\,.
\ee
Canonical transformations acts by conjugation on this matrix, $D \,\arr\, \Omega D \,{}^{t}\Omega$.
%
The Casimir of the $\sl_{2}$  algebra is the determinant of the $D$ matrix.
\be
\cD=\det D= \ell_{-}\ell_{+}-\ell_{0}^{2}=j_{3}^2-k_{1}^2-k_{2}^2\,,
\qquad
\{\cD,\ell_{\alpha}\}=0\,.
\ee
Since the Poisson brackets of the angular momentum with the scalar product observables vanish, $\{J^{a},\ell_{\beta}\}=0$, the canonical transformations $\Omega\in\Sp(2)=\SL(2,\R)$ actually commute with the 3d rotations $\cO\in\SO(3)$.
Moreover these two sets of observables satisfy a Casimir balance equation:
\be
\vJ^{2}=|\vx\w\vp|^{2}=\vx^{2}\vp^{2}-(\vx\cdot\vp)^{2}=\det D\,.
\ee

Keeping in mind that the two vectors $\vx$ and $\vp$ are in fact the integrated components of the triad on the surface, $\ve_{1}$ and $\ve_{2}$,  the $\ell$'s observables encode the surface intrinsic metric data: the matrix $D$ is the integrated induced 2d metric $q$ on the patch and the determinant $\cD$ is the squared density factor $(\det\,q )$. The vector $\vJ$ is the normal vector to the surface patch and encodes the extrinsic curvature describing the embedding of the surface into 3d space. In this approach, the normal vectors are truly angular momenta, as in the spinning geometry interpretation described in \cite{Freidel:2013bfa}.
The area of the patch can be derived from either the intrinsic surface geometry or the extrinsic geometry as $\cA=|\vJ|=\sqrt{\det D}$, with this Casimir balance equation playing the role of a Gauss-Codazzi equation expressing the compatibility of the intrinsic and extrinsic geometries of the surface.

Finally, as the vectors encode the 2d metric, the $\SL(2,\R)$ transformations, realizing canonical transformations of the pair of vectors $(\vx,\vp)$, are to be understood as area-preserving diffeomorphisms of the discretized bubble surface.

\subsection{Gluing bubbles and the bubble network phase space}

Once we have the phase space structure for each bubble, we would like to consider a network of bubbles glued with each other forming the 3d space. We describe the combinatorics of the bubble network by its dual 1-skeleton, introducing the graph $\Gamma$ whose nodes represent the bubbles and whose edges link  pairs of bubble glued with each other. Every edge is thus dual to a surface patch of the two corresponding glued bubbles. We will consider a compact 3d space, corresponding a closed graph $\Gamma$.

Let us now dress this graph with the surface variables introduced above encoding the 2d boundary geometry of the bubbles. This leads to a network of vectors. Around each vertex $v\in\Gamma$, that is for each bubble, we dress each edge $e$ attached to the vertex $v\in e$ with  a pair of vectors $(\vx^{v}_{e},\vp^{v}_{e})\in(\R^3)^{\times 2}$. In order to properly define the symplectic form and transports between the bubbles, it is convenient to orient the graph $\Gamma$. 
Each edge $e$ thus has one canonical vector pair attached to its source vertex $s(e)$ and one attached to its target vertex $t(e)$.
%
%
We need flip to the symplectic structure at the target of every edge $e$:
\be
\{(x^{s}_{e}){}^{a},(p^{s}_{e}){}^{b}\}=\delta^{ab}\,,
\qquad
\{(x^{t}_{e}){}^{a},(p^{t}_{e}){}^{b}\}=\,-\delta^{ab}\,.
\ee
%
This sign flip corresponds to exchanging the role of the position and conjugate momentum, which corresponds geometrically to an orientation flip (switching between the interior and exterior of the bubble).
We then impose a vector  constraint around each vertex and a symplectic constraint along each edge:
\begin{itemize}
\item {\it Around each vertex $v$}, one considers the angular momentum vectors of each of the particles and defines the closure constraint:
\be
\sum_{e\ni v} \eps^{e}_{v}\vJ^{e}_{v}=0\,,\qquad
\vJ_{e}^v=\vx_{e}^v\w\vp_{e}^v
\,,
\ee
where the sign $\eps^{e}_{v}=\pm$ registers the relative orientation of the edge with respect to the vertex, positive if the edge is outgoing  $v=s(e)$ and negative if the edge is incoming $v=t(e)$. This constraint generators generate $\SO(3)$ gauge transformations simultaneously rotating all the vectors  around each vertex.

\item {\it Along each edge $e$}, one considers the symplectic generators $\ell$, or equivalently the Gram matrix $D$, of both pairs of vectors at the source and target vertices of the edge and defines the matching constraint:
\be
\forall \alpha\,,\quad
D^{s}_{e}=D^{t}_{e}\,,
\ee
meaning that the norms and scalar product of the  vectors $\vx$ and $\vp$ at both ends of the edge must match. As illustrated on fig.\ref{fig:gluing}, this simply amounts to matching the 2d geometry of the surface patch dual to the edge from the viewpoint of the two bubbles sharing it. 
In particular, the determinant of the $D$-matrices must match, i.e. $|\vJ_{e}^s|=|\vJ_{e}^t|$ which is the standard matching constraint for twisted geometries. Here we introduce a more general symplectic matching for each edge. This matching constraint generates a $\SL(2,\R)$-gauge invariance along each edge, which is physically interpreted as a gauge invariance of the bubble network under 2d surface diffeomorphisms.

\end{itemize}
\begin{figure}[h!]
\begin{tikzpicture}[scale=2.5]
\coordinate (A1) at (0,0);
\coordinate (A2) at (0.8,0.35);
\coordinate (A3) at (.2,.9);

\coordinate (A1a) at (-.3,-0.07);
\coordinate (A1b) at (.07,-0.3);
\coordinate (A2a) at (1.1,0.28);
\coordinate (A3a) at (-.1,.92);
\coordinate (A3b) at (.3,1.2);

\coordinate (B1) at (3,0);
\coordinate (B2) at (2.2,0.35);
\coordinate (B3) at (2.8,.9);

\coordinate (B1a) at (3.3,-0.07);
\coordinate (B1b) at (2.93,-0.3);
\coordinate (B2b) at (1.9,0.3);
\coordinate (B2a) at (2.02,0.05);
\coordinate (B3a) at (2.9,1.2);

\draw (A1) -- (A2) -- (A3) -- (A1);
\draw (B1) -- (B2) -- (B3) -- (B1);

\draw[dashed] (A1a) -- (A1) -- (A1b);
\draw[dashed] (A2a) -- (A2);
\draw[dashed] (A3a) -- (A3) -- (A3b);
\fill[fill=black,fill opacity=0.05] (A1a) -- (A1) -- (A1b);
\fill[fill=black,fill opacity=0.05] (A1a) -- (A1)  -- (A3) -- (A3a);
\fill[fill=black,fill opacity=0.05] (A3a) -- (A3) -- (A3b);
\fill[fill=black,fill opacity=0.05] (A3b) -- (A3) -- (A2) -- (A2a);
\fill[fill=black,fill opacity=0.05] (A2a) -- (A2) -- (A1) -- (A1b);

\draw[dashed] (B1a) -- (B1) -- (B1b);
\draw[dashed] (B3a) -- (B3);
\draw[dashed] (B2a) -- (B2) -- (B2b);
\fill[fill=black,fill opacity=0.05] (B1a) -- (B1) -- (B1b);
\fill[fill=black,fill opacity=0.05] (B1a) -- (B1)  -- (B3) -- (B3a);
\fill[fill=black,fill opacity=0.05] (B3a) -- (B3) -- (B2) -- (B2b);
\fill[fill=black,fill opacity=0.05] (B2b) -- (B2) -- (B2a);
\fill[fill=black,fill opacity=0.05] (B2a) -- (B2) -- (B1) -- (B1b);

\coordinate (O) at (.25,.3);
\coordinate (O1) at (.25,.6) ;
\coordinate (O2) at (.44,.4);
\draw [->,>=latex] (O)--(O1) node[right]{$\ve_{1}^{\,s}$};
\draw [->,>=latex] (O)--(O2)  node[right]{$\ve_{2}^{\,s}$};

\coordinate (P) at (2.8,.3);
\coordinate (P1) at (2.8,.6) ;
\coordinate (P2) at (2.6,.4);
\draw [->,>=latex] (P)--(P1) node[left]{$\ve_{1}^{\,t}$};
\draw [->,>=latex] (P)--(P2)  node[left]{$\ve_{2}^{\,t}$};

\centerarc[->,>=stealth](1.5,-1.8)(105:75:2.5);
\coordinate (E) at (1.55,.83);
\draw (E) node {$h_{e}\in\SO(3)$};
\coordinate (E) at (1.5,.57);
\draw (E) node {$D_{e}^s=D_{e}^t$};

\end{tikzpicture}
\caption{Gluing of two bubbles imposing the matching of the 2d geometry of the corresponding surface patches through the symplectic constraints $D_{e}^s=D_{e}^t$, equating the norms $|\ve_{1}^{\,s}|=|\ve_{1}^{\,t}|$, $|\ve_{2}^{\,s}|=|\ve_{2}^{\,t}|$ and the scalar product $\ve_{1}^{\,s}\cdot \ve_{2}^{\,s}=\ve_{1}^{\,t}\cdot \ve_{2}^{\,t}$, resulting in the existence of a unique $\SO(3)$ transport between the two bubbles given by the group element $h_{e}$.
\label{fig:gluing}}
\end{figure}
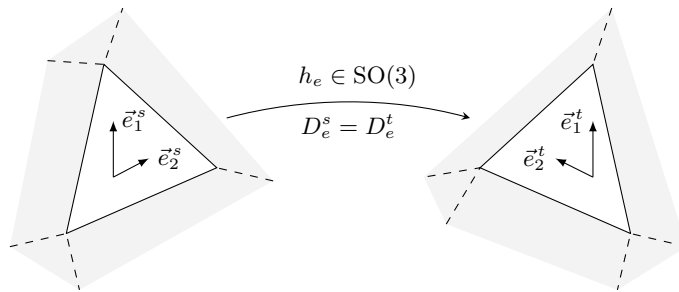

At this stage, the bubble networks is dressed only with the discretized veirbein on the surfaces, given by the pairs of vectors $(\vx_{e}^v,\vp_{e}^v)$.  These variables not only encode the intrinsic 2d geometry of the surfaces, but we can further recover the 3d transport between bubbles as $\SO(3)$ group elements along the network's edges.
Indeed, assuming the matching conditions given above, we can reconstruct  a unique $\SO(3)$ group element mapping the canonical pair of vectors $(\vx^{s}_{e},\vp^{s}_{e})$ at the edge's source vertex to the other pair $(\vx^{t}_{e},\vp^{t}_{e})$ living at at its target.
Let us drop the index $e$ for this analysis. The matching conditions impose that the norms and scalar product of the position and momentum at the source equal those at the target:
\be
\ell_{-}
=|\vx {}^{s}|^{2}=|\vx\,{}^{t}|^{2}\,,\quad
\ell_{+}=|\vp\,{}^{s}|^{2}=|\vp\,{}^{t}|^{2}\,,\quad
\ell_{0}=\vx\,{}^{s}\cdot\vp\,{}^{s}=\vx\,{}^{t}\cdot\vp\,{}^{t}\,,
\quad
\cD=(\ell_{-}\ell_{+}-\ell_{0}^{2})=|\vx\,{}^{s}\w\vp\,{}^{s}|^{2}=|\vx\,{}^{t}\w\vp\,{}^{t}|^{2}\,,
\ee
so that, as long as $\cD\ne 0$ (ensuring that the coordinate and momentum are not collinear), there exists a unique rotation $h_{e}$ along the edge $e$ mapping the pair of vectors at the source onto their target counterpart. This $\SO(3)$ holonomy is explicitly given by:
\be
\left|
\begin{array}{l}
h_{e}\act\vx\,{}^{s}=\vx\,{}^{t}\,, \\
h_{e}\act\vp\,{}^{s}=\vp\,{}^{t}\,,
\end{array}
\right.
\qquad
(h_{e})_{ab}
\,=\,
\f1{\cD}\Big{[}
\ell_{+}x_{a}^tx_{b}^s+\ell_{-}p_{a}^tp_{b}^s
-\ell_{0}(p_{a}^tx_{b}^s+x_{a}^tp_{b}^s)
+(x^t\w p^t)_{a}(x^s\w p^s)_{b}
\Big{]}\,,
\ee
with averaged expressions for all the norm factors on the edge:
\be
\cD=|x^s\w p^s|\,|x^t\w p^t|\,,
\quad
\ell_{+}=|p^s|\,|p^t|\,,
\quad
\ell_{-}=|x^s|\,|x^t|\,,
\quad
\ell_{0}=\sqrt{(x^s\cdot p^s)\,(x^t\cdot p^t)}\,.
\ee
The matrix $h_{e}$ maps the vector $\vx\,{}^s$ to $\vx\,{}^t$, the vector $\vp\,{}^s$ to $\vp\,{}^t$, and the direction $(\vx\,{}^s\w\vp\,{}^s)$ to $(\vx\,{}^t\w\vp\,{}^t)$. Since the the scalar products $x^s\cdot p^s$ and  $x^t\cdot p^t$ are equal, this is enough to ensure that $h_{e}$ is an orthogonal matrix (see in appendix \ref{app-SO3holonomy} for more details).

Modifying the relative weight of the source and target  factors in the formulae above does not change their actual value once the symplectic matching conditions are enforced. This would nevertheless affect the Poisson brackets of the holonomy $h_{e}$. To obtain the correct Poisson brackets and recover twisted geometries as we show in the next section, the averaged choice given above ensures that the holonomy commutes with itself and that the symplectic reduction by the matching conditions works smoothly.

\medskip

Jumping ahead to the quantum level, imposing constraints between the two systems living on an edge of the network essentially creates entanglement along the edge.
Here we can understand the symplectic matching constraints as leading to, in quantum information terms, a maximal entanglement along each edge $e$, i.e. at the 2d interfaces between bubbles, compatible with the symmetry group structure.
Indeed, an unconstrained edge would carry uncorrelated pairs of vectors at its source and target. Written in term of quantum vectors, such a ``naked edge'' would be represented\footnotemark{} by a pair of states $|v^s\ra\otimes |v^t\ra$, where $v^s$ and $v^t$ respectively encode the classical data $(x^s,p^s)$ and $(x^t,p^t)$ at the source and target of the edge. Each end of the edge carries an action of the $\SO(3)\times\SL(2,\R)$ Lie group.
Imposing the $\sl(2,\R)$ symplectic matching constraints  then relates the states at the source and target of the edge by a $\SO(3)$ transformation: a basis of states solving the constraint are  maximally entangled states,  $\int \dd k \, k|v\ra\otimes hk|v\ra$, depending on the rotation $h\in\SO(3)$ and defined by a group averaging over $\SO(3)$.
One could then glue bubble boundary states using those maximally entangled states along the edges.
The precise quantization of the bubble network phase space should nevertheless be carried out in detail in order to realize this intuition explicitly.
\footnotetext{
To be more precise, let us sketch a quantization scheme  in terms of coherent states. We use the Segal-Bargmann representation for the pair of conjugate vectors $(\vx,\vp)\in(\R^3)^{\times 2}$. For $i$ running from 1 to 3, we quantize each vector component $(x_{i},p_{i})$ as a harmonic oscillator and represent them at the quantum level as acting on holomorphic wave-functions $\phi(z_{i})$, with $z_{i}$ being the label of the coherent state and the annihilation (resp. creation) operator represented as the multiplication operator $a_{i}=z_{i}$ (resp. the derivation operator $a_{i}^\dagger=\pp_{z_{i}}$. Group transformations in $\SO(3)$ act as 3d rotations on the complex vector $(z_{1},z_{2},z_{3})$, while the $\sl(2,\R)$ algebra is generated by the total energy $\sum_{i}(a_{i}^\dagger a_{i}+1/2)$ and the squeezing operators $\sum_{i} a_{i}^2$ and $\sum_{i} a_{i}^\dagger{}^2$:
\be
\left[\sum_{i}a_{i}^\dagger a_{i}, \sum_{i} a_{i}^2\right]=-2\sum_{i} a_{i}^2
\,,\qquad
\left[\sum_{i}a_{i}^\dagger a_{i}, \sum_{i} a_{i}^\dagger{}^2\right]=+2\sum_{i} a_{i}^\dagger{}^2
\,,\qquad
\left[\sum_{i} a_{i}^\dagger{}^2,\sum_{i} a_{i}^2\right]
=
-2\sum_{i}(a_{i}^\dagger a_{i}+\f12)
\,.
\nn
\ee
Considering an edge, we have two copies of this structure, one at its source in terms of coherent state label $z_{i}$ with operators $a_{i},a_{i}^\dagger$ and one at its target in terms of label $w_{i}$ with operators $b_{i},b_{i}^\dagger$.
The $\sl(2,\R)$ matching constraints are:
\be
\sum_{i}z_{i}\pp_{z_{i}}=\sum_{i}w_{i}\pp_{w_{i}}
\,,\qquad
\sum_{i}z_{i}^2=\sum_{i}\pp_{w_{i}}^2
\,,\qquad
\sum_{i}\pp_{z_{i}}^2=\sum_{i}w_{i}^2
\,.
\nn
\ee
It is straightforward to check that a basis of solutions to these constraints is given by the entangled states  $\cosh[z_{i}h_{ij}w_{j}]$ and $\sinh[z_{i}h_{ij}w_{j}]$  labeled by a group element $h\in\SO(3)$, which are exactly the even and odd superpositions of all coherent states at the source and target such that the two states differ that the given rotation $h$.
%
In order to realize the explicit quantization of the bubble network phase space in terms of extended spin networks, we would need to refine this analysis using irreducible representations of the symmetry group  $\SU(2)\times\SL(2,\R)$.
}

\section{From Bubble Networks back  to Twisted Geometries}

\subsection{Symplectic reduction by the $\sl(2,\R)$ matching constraints}

Let us put aside the closure constraints, generating the local gauge invariance under 3d rotations at every node of the network, and focus on the matching constraints along the edges. We showed above how to reconstruct $\SO(3)$ group elements describing the transport between bubbles along the graph edges. Here we actually show that if we take into account the $\sl(2,\R)$ matching constraints and quotient the vector phase space by the associated $\SL(2,\R)$ gauge transformations, we exactly recover the $T^*\SO(3)$ phase space of twisted geometries. The closure constraints at the vertices do not play any role at this stage and can be imposed a posteriori without interfering with this symplectic reduction.

\begin{prop}
The symplectic quotient of the vector network phase space $(\R^{6})^{\times 2E}$ by the symplectic matching constraints is the twisted geometry phase space $T^{*}\SO(3)^{E}$ parametrized by $\vJ^{s,t}_{e}\in\R^{3}$ and $h_{e}\in\SO(3)$:
\be
(\R^{6})^{\times 2E}//\,\SL(2,\R)^{E}\sim T^{*}\SO(3)^{E}
\,.
\ee
\end{prop}
\begin{proof}
We start with a simple dimension counting:
\be
\dim (\R^{6})^{\times 2E} -2\dim\SL(2,\R)^{E}
=6E
=\dim T^{*}\SO(3)^{E}\,.
\ee
Next, the angular momenta $\vJ^{s,t}_{e}$ and the $\SO(3)$-holonomies $h_{e}$ are both invariant under symplectic transformations on each edge and Poisson-commute with the matching constraints $(D_{e}^{s}-D_{e}^{t})=0$. Finally, a straightforward though lengthy calculation allows to check that their Poisson brackets (weakly) satisfy  the brackets of the $T^{*}\SO(3)^{E}$ phase space once the matching constraints are imposed:
\be
\{(J^{s}_{e})_{a},(J^{s}_{e})_{b}\}= \eps_{abc}(J^{s}_{e})_{c}
\,,
\quad
\{(J^{t}_{e})_{a},(J^{t}_{e})_{b}\}= -\eps_{abc}(J^{t}_{e})_{c}
\,,
\ee
\be
\{h_{e},h_{e}\}\sim 0
\,,
\quad
\{\vJ^{s}_{e},h_{e}\}\sim h_{e}\vcJ
\,,
\quad
\{\vJ^{t}_{e},h_{e}\}\sim -\vcJ h_{e}
\,,
\ee
where the $\cJ$'s are the $\so(3)$-generators, defined as $3\times3$ matrices:
$$
(\cJ^{a})_{bc}=\eps_{abc}\,.
$$
So that the source angular momentum generates 3d rotations on the right of the holonomy $h_{e}$ while the target angular momentum generates rotations on the left.

\end{proof}

\subsection{Map to twisted geometries and twist angle}

Focusing on a single surface patch, we  started with the phase space parameterized by a canonical pair of 3-vectors $(x^{a},p^{a})$ with a total of 6 real independent variables. We mapped them on  another pair of 3-vectors $(J^{a},\ell_{\alpha})$, which now commute with each other, $\{J^a,\ell_{\alpha}\}=0$. From a practical point of view, the map down to twisted geometries from bubble networks is to drop the $\ell$'s and focus on the $J$'s. There is however an important subtlety.
The pair of vectors $(J^{a},\ell_{\alpha})$ satisfy the Casimir balance equation. This means that they encode only 5 independent parameters and we need one extra variable to fully parametrize the phase space and reconstruct the initial coordinate and momentum vectors.

What's missing is the choice of a direction orthogonal to the angular momentum vector $\vJ$. Indeed both coordinate and momentum vectors are orthogonal to the angular momentum, $\vx\cdot\vJ=\vp\cdot\vJ=0$, and we need to specify the direction of at least one of them.
For instance, we can start with $\vJ$ and the $\ell_{\alpha}$, satisfying the balance equation $\vJ^{2}=\det D$, and further specify the direction $\hp\in\cS^{2}$ on the unit sphere with $\hp\cdot \vJ=0$. This is enough to reconstruct the vector $\vp$:
\be
\vp= |\vp|\,\hp
\,=\,
\sqrt{\ell_{+}}\,\hp\,,
\ee
and then recover the coordinate vector $\vx$ by a cross product\footnotemark:
\be
\vp\w\vJ=
\vp\w(\vx\w\vp)=
\vp^{2}\,\vx-(\vp\cdot\vx)\vp
=\ell_{+}\vx-\ell_{0}\vp\,
\Rightarrow
\quad
\vx=\f1{\ell_{+}}{(\vp\w\vJ+\ell_{0}\vp)}\,.
\ee
\footnotetext{
This reconstruction of the position vector is actually very similar to the definition of position Dirac observables for a relativistic particle \cite{Freidel:2007qk}.
}
This reconstruction only works when $\ell_{+}=\vp^{2}$ does not vanish. Otherwise, if the momentum vanishes, $\vp=0$, then the angular momentum vanishes too $\vJ=0$ and the whole reconstruction issue becomes degenerate.

Specifying the direction $\hp\perp\vJ$ is equivalent to specifying an angle $\theta\in[0,2\pi]$ in the plane orthogonal to $\vJ$. We choose for example the $z$-direction, then the vector $\he_{z}\w\vJ$ lays in the plane orthogonal to $\vJ$ and we can choose $\theta$ to be the angle\footnotemark{} between this reference direction and $\hp$.
\footnotetext{
More technically, we would normalize $\he_{z}\w\vJ$ to define $\hv_{x}$ and define the third direction of this orthonormal frame as $\hv_{y}=\hJ\w\hv_{x}$, then
$$
\hp=\cos\theta\hv_{x}+\sin\theta\hv_{y}\,.
$$
}
This gives a bijection between the pair of vectors $(x^a,p^a)$ and the variables $(J^a,\ell_{\alpha},\theta)$ satisfying the balance equation $\vJ^2=\vec{\ell}^2$.

\medskip

This angle $\theta$ allows to recover the twist angle of the twisted geometry interpretation.
More precisely, the bubble network data on a graph $\Gamma$ consist in dressing each graph edge $e$ with two pairs of vectors $(\vx_{e}^{s},\vp_{e}^{s})$ and $(\vx_{e}^{t},\vp_{e}^{t})$, living at at its two extremities. For each half-edge, we define the angular momentum   $\vJ^{v}_{e}=\vx_{e}^{v}\w\vp_{e}^{v}$ and the  symplectic observables $\vec{\ell}^v_{e}$ (or equivalently the matrix $D^v_{e}$  encoding the 2d surface metric data), which satisfy the balance equation $(\vJ^{v}_{e})^2=(\vec{\ell}^v_{e})^2$. Then, assuming that the vectors satisfy the symplectic matching constraints, $\vec{\ell}^s_{e}=\vec{\ell}^t_{e}$, they allow to define  a $\SO(3)$ group element $h_{e}$.

To make the link between the bubble networks and the twisted geometries, we keep  the angular momentum vectors for each half-edge $\vJ^{v}_{e}$ and make the $\ell$'s aside. From the twisted geometry point of view, the $\SO(3)$ holonomy $h_{e}$ sends $\vJ^{s}_{e}$ onto $\vJ^{t}_{e}$, but is not fully determined by these two vectors. We require the extra data of a twist angle $\varphi_{e}$ to reconstruct a unique $\SO(3)$ parallel transport along the edge. From the new perspective of the bubble networks, as we have explained above, we require the extra data of an angle $\theta^{s}_{e}$ at the source vertex to reconstruct the pair of vectors $(\vx_{e}^{s},\vp_{e}^{s})$ from the angular momentum and symplectic observables $(\vJ^{s}_{e},\vec{\ell}^s_{e})$, and similarly at the target vertex. Then  the unique $\SO(3)$-holonomy mapping $(\vx_{e}^{s},\vp_{e}^{s})$ to $(\vx_{e}^{t},\vp_{e}^{t})$ not only depends on $\vJ^{s}_{e}$ and $\vJ_{e}^{t}$ but also on those angles $\theta^{s,t}_{e}$. Comparing these two points of view, the twist angle is simply the difference $\varphi_{e}=\delta\theta_{e}=(\theta^{t}_{e}-\theta^{s}_{e})$.

\bigskip

Let us clarify the hierarchy of geometrical structures from Regge geometries to the extended twisted geometries defined from the bubble network phase space introduced in the present work. Twisted geometries are defined as networks of holonomies and fluxes on a graph and are understood as extensions of Regge geometries. Indeed, in general, one can reconstruct a polyhedron dual to each node of a twisted geometry, with each edge attached to that node being dual to one of the faces of the polyhedron \cite{Bianchi:2010gc}. So every edge is understood as linking two polyhedra with the $\SU(2)$ holonomy along that edge encoding the change of frame from one polyhedron to the next. Then we impose an area-matching constraint  across every edge, equating the area of the faces of the two neighboring polyhedra. Nevertheless this is not a geometric gluing and the shape of the two faces do not necessarily match.
We go further and be more precise if we restrict ourselves to simplicial 3d geometries, i.e. triangulations, and thus to 4-valent networks. In that case, one can introduce gluing constraints between tetrahedra, on top of the area-matching constraints, that enforce the shape-matching of triangles and not only the matching of their areas \cite{Dittrich:2008ar,Dittrich:2010ey}. This allows to recover Regge triangulations as a special case of twisted triangulations, when gluing constraints are imposed.

The non-shape-matching of triangles for twisted geometries were further interpreted in \cite{Haggard:2012pm} as allowing for torsion. More precisely, considering a triangle defined by two edge vectors $\vv_{1}$ and $\vv_{2}$, one can deform its shape without changing its area by doing a $\SL(2,\R)$ transformation on this pair of vectors. Generic twisted geometries allow for such a deformation between two neighboring tetrahedra. This deformation must be taken into account in the definition of a discretized 3d spin-connection, which is not simply defined by the $\SU(2)$ holonomies. This underlines the difference between the torsionless 3d spin-connection and the Ashtekar-Barbero connection (used to define the $\SU(2)$ holonomies) which has a non-trivial torsion related to the extrinsic curvature of the 3d slice. When the $\SL(2,\R)$ transformations are frozen and the shape-matching of triangles is imposed by requiring that the scalar products $\vv_{A}\cdot\vv_{B}$ for $A,B\in\{1,2\}$ match on both ends of each edge, then we recover Regge triangulations from twisted triangulations. Finally, these shape-matching constraints can be entirely written in terms of the triangle normal vectors -the fluxes- in the case of triangulations (but this does not work as easily for generic cellular decomposition).

Let us see how bubble networks fit in this picture.
Compared to twisted geometries, we add extra data to each face around every node, introducing two frame vectors $\ve_{1}$ and $\ve_{2}$ instead of only the face normal vector $\vN$. The normal vector is recovered as $\vN=\ve_{1}\w\ve_{2}$, but we also have access to the 2d metric on the face, $g_{AB}=\ve_{A}\cdot \ve_{B}$. We upgrade the area-matching constraint across edges, equating the norm of the normal vector on both ends of each edge $|\vN^s|=|\vN^t|$, to a $\sl_{2}$-matching constraint $g_{AB}^s=g_{AB}^t$. Formulated as such, bubble networks look very similar to twisted geometries with the shape matching constraints, but they are not. To start with, bubble networks are introduced to be able to account for non-flat boundary surfaces with non-trivial 2d metric, thereby generalizing both Regge geometries and twisted geometries. More precisely, the identification of the shape matching constraints to the symplectic matching constraints, in the restricted case of triangulations, would rely on the identification of the edge vectors of the triangulation with the frame vectors, $\ve_{A}=\vv_{A}$ in the notations above. Not only this means providing the triangulation edge vectors with a Poisson bracket  (and choosing a root vertex\footnotemark{} for each triangle to select the two vectors $\vv_{1}$ and $\vv_{2}$ as the frame vectors entering the matching constraints), it also implicitly means that we are working with flat faces and that we define the surface frame fields on each face from the 1d data living on each face boundary. Note that, in the present work and the bubble network framework, we have not yet considered the algebraic and geometric data carried by 1-cells but focused instead of the geometric data carried by 3-cells and 2-cells.
\footnotetext{ Considering a triangle made of three edges, $\vv_{1,2,3}$ satisfying a closure condition $\vv_{1}+\vv_{3}=\vv_{2}$, with the normal vector defined as $\vN=\vv_{1}\w\vv_{2}=\vv_{1}\w\vv_{3}=\vv_{2}\w\vv_{3}$, we can define a Poisson bracket:
\be
\{v^a_{1},v^b_{2}\}=\{v^a_{1},v^b_{3}\}=\{v^a_{2},v^b_{3}\}=\delta^{ab}\,.
\nn
\ee
If we choose the pair of  vectors $(\vv_{1},\vv_{2})$ as frame fields, then change root vertex and switch to the pair of vectors $(\vv_{1},\vv_{3})$, this is a simple canonical transformation realized as a $\SL(2,\R)$ transformation.
}

So, we have two complementary perspectives. On the one hand, if the case of triangulations, if we identify the face frame vectors to the triangulation edge vectors (if this can be done in a consistent way), then the $\sl_{2}$-matching constraints of bubble networks coincide with the shape-matching constraints of twisted geometries and we recover Regge triangulation directly from bubble networks. This is due to the fact that $\SL(2,\R)$ transformations allows to explore the whole space of triangle shapes  at fixed area.
On the other hand, we should include the 1d graph structure drawing the contours of the faces on the 2d boundary of every 3-cells, and understand the geometric and algebraic data carried by these 1-cells, in order to study in which case we can reconstruct the frame vectors from the 1d data and to characterize more precisely in which situations we can identify the frame vectors to the triangulation edge vectors. This is exactly a question to investigate in the framework of spinning geometries \cite{Freidel:2013bfa,Charles:2016xzi}, in which the normal vectors (the fluxes) to the faces are constructed as holonomies of a specific connection along the 1d boundary of those faces.

At the end of the day, bubble networks are twisted geometries extended with the extra data of frame fields, and the reconstruction of frame vectors for twisted geometries, proposed in \cite{Haggard:2012pm}, seems to be a particular case of bubble networks. To make this more precise, we would to consider the 1-skeleton of the 2d cellular decomposition on each bubble and clarify which data is carried by 1-cells, which symplectic structure are they endowed with and how they fit with the 2d metric data of bubble networks.

\subsection{Conformal gauge and spinor parametrization}


Above, we have showed how to reconstruct  $\SO(3)$ holonomies from the vector phase space of bubble networks and how to recover a $T^*\SO(3)$ phase space. To truly recover twisted geometries, we would like to reconstruct $\SU(2)$ holonomies and recover a $T^*\SU(2)$ phase space. An efficient way to do so is to directly recover the spinorial phase space for twisted geometries \cite{Freidel:2010bw,Borja:2010rc,Livine:2011gp} from the present construction. More precisely, we simply need to show how to define spinors from pairs of vectors on each half-edge.

Let us start with the data encoded in a spinor. From e.g. \cite{Dupuis:2011wy}, a spinor in $\C^{2}$ is equivalent to  an orthonormal basis in $\R^{3}$.
Indeed, starting from a single spinor $z$, we consider  the set of real quadratic combinations of its components:
\be
\cN = \f12\la z | z\ra\,, \quad
\vJ(z)=\f12\la z|\vsigma|z\ra\,,\quad
\vK(z)=\f14\Big{(}\la z|\vsigma|z]+[z|\vsigma|z\ra\Big{)}\,,\quad
\vL(z)=\f i4\left(\la z|\vsigma|z]-[z|\vsigma|z\ra\right)\,.
\ee
Together they form a closed $\so(3,2)$ algebra under Poisson bracket\footnotemark:
\be\nn
\{J_a(z),J_b(z)\}=\epsilon_{abc}J_c(z),\qquad
\{J_a(z),K_b(z)\}=\epsilon_{abc}K_c(z),\qquad
\{J_a(z),L_b(z)\}=\epsilon_{abc}L_c(z),
\ee
\be\nn
\{K_a(z),K_b(z)\}=-\epsilon_{abc}J_c(z),\qquad
\{L_a(z),L_b(z)\}=-\epsilon_{abc}J_c(z),\qquad
\{K_a(z),L_b(z)\}=\delta_{ab}\cN,
\ee
\be\nn
\{\cN,J_a(z)\}=0,\qquad
\{\cN,K_a(z)\}=-L_a(z),\qquad
\{\cN,L_a(z)\}=+K_a(z).
\ee
\footnotetext{
This algebra can be derived from the following Poisson brackets,
$$
\{\f12\la z|\sigma_a|z\ra,\la z|\sigma_b|z]\}
\,=\,
\f{-2i}2\la z|\sigma_a\sigma_b|z]
\,=\,
\f{-2i}2\la z|\delta_{ab}\id+i\epsilon_{abc}\sigma_c|z]
\,=\,
\epsilon_{abc}\la z|\sigma_c|z],
$$
as well as
$$
\{\f12\la z|\sigma_a|z\ra,[ z|\sigma_b|z\ra\}
\,=\,
\epsilon_{abc}[ z|\sigma_c|z\ra\,,\qquad
\{\la z|\sigma_a|z],[ z|\sigma_b|z\ra\}
\,=\,
4i\delta_{ab}\la z|z\ra-4\epsilon_{abc}\la z|\sigma_c|z\ra.
$$
}
%
%
Moreover, the three vector generators form an orthonormal basis of $\R^3$:
\be\nn
\vJ^2=\vK^2=\vL^2=\cN^2,\qquad
\vJ\cdot\vK=\vJ\cdot\vL=\vK\cdot\vL=0\,.
\ee

Coming back to our vector phase space, it is natural to seek an identification of the triplet $(\vx,\vp,\vJ)$ with the orthonormal basis $(\vK,\vL,\vJ)$. However, this requires that $\vx$ and $\vp$ be orthogonal and with equal norm. It is indeed always possible to reach such a configuration by a canonical $\SL(2,\R)$ transformation. From a 2d geometry point of view, this amounts to using a 2d diffeomorphism to reach the conformal gauge. This means using isothermal coordinates\footnotemark{} for every surface patch of the bubbles.
\footnotetext{
An intriguing remark is that isothermal coordinates for minimal surfaces allow for the Weierstrass-Enneper representation in terms of holomorphic coordinates \cite{minimal}. This might open the door to a direct link between spinning geometries (whose boundary surfaces are all minimal surfaces) and spinor networks.}
More precisely, we proceed to a gauge-fixing of the $\SL(2,\R)$ gauge transformations by imposing two constraints, $\ell_{0}=\vx\cdot\vp=0$ and $\ell_{-}-\ell_{+}=\vx^{2}-\vp^{2}=0$. 
In this basis the 2d metric is proportional to the flat metric 
$q_{AB} = |\vJ| \delta_{AB}$., this is the discrete conformal gauge.

We can then compute the Dirac bracket $\{\cdot,\cdot\}_{D}$, in terms of the Dirac matrix whose only matrix element is $\{\ell_{0},(\ell_{+}-\ell_{-})\}=(\ell_{+}+\ell_{-})$,
and obtain the Dirac bracket $\{\cdot,\cdot\}_{D}$:
\be
\{F,G\}_{D}=
\{F,G\}
+\f1{4|\vJ|}\Big{\{}F,(x^{2}-p^{2})\Big{\}}\,\Big{\{}\vx\cdot\vp,G\Big{\}}
-\f1{4|\vJ|}\Big{\{}F,\vx\cdot\vp\Big{\}}\,\Big{\{}(x^{2}-p^{2}),G\Big{\}}
\,,
\ee
with $|\vJ|=x^{2}=p^{2}$.
We easily find that  $(|\vx|\vx,|\vp|\vp,\vJ)$ exactly reproduces the $\so(3,2)$ algebra of the $(\vK,\vL,\vJ)$ generators given above. The norm factors are here to ensure that we do indeed have three vectors with equal norm (thus forming an orthonormal basis).

From here, we are ready to reconstruct the spinor from the pair of vectors.  Indeed, it is not possible to fully recover the original spinor $z\in\C^{2}$ from only the angular momentum vector $\vJ\in\R^{3}$. These miss the information of the spinor phase $e^{i\vphi}$ (which corresponds to the twist angle data):
$$
|z\ra
\,=\,
e^{i\vphi}\mat{c}{\sqrt{\cN+J_{3}}\\ e^{i\theta}\sqrt{\cN-J_{3}}}\,,\qquad
e^{i\theta}=\f{J_{1}+iJ_{2}}{\sqrt{J_{1}^{2}+J_{2}^{2}}}=\f{J_{+}}{\sqrt{\cN^{2}-J_{3}^{2}}}
\,,\qquad
e^{i\vphi}\in\U(1)\,.
$$
On the other hand, once the whole orthonormal triad $(\vK,\vL,\vJ)$ is provided, we can retrieve the whole spinor with its phase information and express its components $z_{A=0,1}$, say, in terms of the $\vK$ and $\vL$ vectors:
\be
(z_{0})^{2}=K_{-}+iL_{-}\,,\qquad
(z_{1})^{2}=-\big{(}K_{+}+iL_{+}\big{)}\,,
\ee
with $K_{\pm}=K_{1}\pm i K_{2}$ and similarly $L_{\pm}=L_{1}\pm i L_{2}$. One simply needs to take care of choosing an appropriate cut for the square-root on the complex plane.
From the $\so(3,2)$ commutators, it is easy to check that, once $K$ and $L$ are defined in terms of the position and momentum, we get the expected commutators:
\be
\left|
\begin{array}{l}
\vK=|x|\vx=|\vp|\vx \\
\vL=|\vp|\vp=|x|\vp
\end{array}
\right.
\quad\Longrightarrow\quad
\left|
\begin{array}{l}
\{z_{0},z_{1}\}_{D}=0\\
\{z_{0},\bz_{1}\}_{D}=0\\
\{z_{0}^{2},\bz_{0}^{2}\}_{D}=-4i(\cN+J_{3})=-4iz_{0}\bz_{0}\\
\{z_{1}^{2},\bz_{1}^{2}\}_{D}=-4i(\cN-J_{3})=-4iz_{1}\bz_{1}
\end{array}
\right.
\ee
This allows a direct mapping between the  spinor network representation of twisted geometry and our vector parametrization of the bubble network phase space:
\be
\left|
\begin{array}{l}
(z_{0})^{2}=(|\vx|\,p_{2}+|\vp|\,x_{1})\,-\,i\,(|\vp|\,x_{2}-|\vx|\,p_{1})
\\
(z_{1})^{2}=(|\vx|\,p_{2}-|\vp|\,x_{1})\,-\,i\,(|\vp|\,x_{2}+|\vx|\,p_{1})
\end{array}
\right.
\ee
This means that the (holomorphic) spinorial representation of loop quantum gravity is merely a gauge-fixed version of the full phase space presented here. Bubble networks are extended twisted geometries with the extra data of  discretized surface metrics, which reduces to the original twisted geometries in the conformal gauge when choosing isothermal coordinates on the bubbles' surfaces and thus working with orthonormal triads.

\section*{Outlook}

We have introduced discrete bubble networks for the 3d geometry as a discrete version of a manifold atlas, with the charts represented by bubbles carrying algebraic data allowing to glue them into a consistent geometry. And we showed how this leads (in a suitable choice of gauge fixing) to the kinematical structures of loop quantum gravity where the states of 3d geometry are defined as twisted geometries and spin networks.

More precisely, the discrete bubble networks, as illustrated on fig.\ref{fig:bubble}, carry the data of the 2d discretized metric of the bubbles' surfaces and implement a consistent gluing of the bubbles through a matching of the boundary 2d geometry. We have shown that these can be understood as extended twisted geometries, with a dual algebraic structure with $\SU(2)$ group elements describing the 3d transport between the reference frames of each bubble and a local $\SL(2,\R)$ gauge invariance of the gluing interpreted as the discrete equivalent of  surface diffeomorphisms. This implements the discrete network version of the analysis of boundary surfaces coupled to loop gravity worked out in \cite{Freidel:2015gpa}.
The main new ingredient of our framework compared to twisted geometries, is that the standard area-matching constraint between two 3d cells (or equivalent two spin network nodes) is augmented to a $\sl_{2}$-matching constraint, understood as a matching of the 2d metric at the interface. Similar dilatation matching constraints were proposed in a slightly different context in \cite{Langvik:2016hxn} which investigated an action of the conformal group $\SO(4,2)$ on twisted geometries and spin network states turning them into $\su(2,2)$ spin networks.  This was achieved through the extension of the spinors to twistors, which is different from our extension of the flux vector to a triad. The resulting $\so(4,2)$ structure is thus rather different from the $\su(2)\times\sl_{2}$ structure derived here. It should nevertheless be interesting to merge those two extensions together into extended covariant twisted geometries by upgrading the $\SU(2)$ transport between 3d cells to $\SL(2,\C)$ group elements and thereby describing 3d bubble networks embedded in a 4d space-time geometry (with non-trivial extrinsic geometry).
\begin{figure}[h!]
\begin{tikzpicture}[scale=2]

\coordinate (A1) at (0,0);
\coordinate (A2) at (0.5,-0.1);
\coordinate (A3) at (.65,.4);
\coordinate (A4) at (.3,.4);

\draw (A1) -- (A2) -- (A3) -- (A4) -- (A1);
\fill[fill=black,fill opacity=0.05]  (A1) -- (A2) -- (A3) -- (A4);

\coordinate (A5) at (.95,0);
\draw (A2) -- (A5) -- (A3);
\fill[fill=black,fill opacity=0.05]  (A2) -- (A5) -- (A3);

\coordinate (A6) at (1.1,.65);
\draw (A5) -- (A6) -- (A3);
\fill[fill=black,fill opacity=0.05]  (A5) -- (A6) -- (A3);

\coordinate (A7) at (.57,.8);
\draw (A3) -- (A7) -- (A6);
\fill[fill=black,fill opacity=0.05]  (A3) -- (A7) -- (A6);
\draw (A4) -- (A7);
\fill[fill=black,fill opacity=0.05]  (A4) -- (A7) -- (A3);

\coordinate (A8) at (.93,1.05);
\draw (A6) -- (A8) -- (A7);
\fill[fill=black,fill opacity=0.05]  (A6) -- (A8) -- (A7);

\coordinate (A9) at (.62,1.2);
\draw (A8) -- (A9) -- (A7);
\fill[fill=black,fill opacity=0.05] (A8) -- (A9) -- (A7);

\coordinate (A10) at (-0.1,.75);
\draw (A9) -- (A10) -- (A7);
\fill[fill=black,fill opacity=0.05] (A9) -- (A10) -- (A7);
\draw (A10) -- (A4);
\fill[fill=black,fill opacity=0.05] (A4) -- (A10) -- (A7);
\draw (A10) -- (A1);
\fill[fill=black,fill opacity=0.05] (A4) -- (A10) -- (A1);

\coordinate (B1) at (.4,.15);
\draw (B1) node{$\bullet$};
\draw[blue] (B1)--++ (-.1,-.7);
\coordinate (B2) at (.9,.35);
\draw (B2) node{$\bullet$};
\draw[blue] (B2)--++ (.7,-.1);
\coordinate (B3) at (.1,.35);
\draw (B3) node{$\bullet$};
\draw[blue] (B3)--++ (-.7,-.1);
\coordinate (B4) at (.33,.92);
\draw (B4) node{$\bullet$};
\draw[blue] (B4)--++ (-.2,.6);

\coordinate (a1) at (.73,.65);
\draw[dotted] (A8) -- (a1) -- (A5);
\coordinate (a2) at (.33,.6);
\draw[dotted] (A9) -- (a1) -- (a2) -- (A10);
\coordinate (a3) at (.55,.3);
\draw[dotted] (a1) -- (a3) -- (a2) ;
\draw[dotted] (A1) -- (a3) -- (A5) ;

\coordinate (b1) at (.75,1);
\draw[opacity=.2] (b1) node{$\bullet$};
\draw[blue,opacity=.7] (b1)--++ (.23,.6);


\end{tikzpicture}
\caption{In the bubble network framework, a bubble with discretized 2d geometry has the topology of a 3-ball with a 2-sphere boundary and is represented  as a (not necessarily convex) polyhedron. From the perspective of twisted geometries, a bubble is the blown-up version of a graph node. Then the graph edges (here in blue) link each surface patch on the bubble's surface to a neighboring bubble. These edges carry the symplectic matching constraints for the 2d geometry and the $\SO(3)$ transport from bubbles to bubbles.
\label{fig:bubble}}
\end{figure}
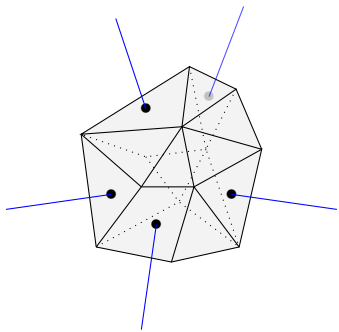

This formalism allows direct access to the 2d metric of boundary surfaces, opening the door to defining boundary geometrical observables in loop quantum gravity such as  2d curvature or quasi-local energy defined as surface integrals (e.g. \cite{Yang:2008th}). Furthermore it allows for conformal transformations of the 2d boundary metric on ``corners''. Indeed, we have showed that the twisted geometries (in their spinorial representation) are gauge-fixed versions of the new bubble network phase space in the conformal gauge for the 2d metric on the bubbles. The extended phase space thus allows to unfreeze this gauge-fixing and explore the whole phase space of boundary 2d geometries, which was inaccessible in the standard loop quantum gravity formalism.
This improvement should be very useful in the investigation of  the possible holography of the dynamics of quantum geometry in loop quantum gravity, for example through a quasi-local CFT/gravity duality.

In the meanwhile, we foresee a few possible extensions of the present formalism and interesting outlook:

\begin{itemize}

\item {\it Quantization \& extended spin networks:}
A first task would be to quantize our extended twisted geometries into extended spin networks, as graphs dressed with $\SU(2)$ representations and invariant tensors augmented with $\sl_{2}$ charges. Each node of the graph would be dual to a bubble with a real 2d metric on its boundary surface. A priori, the bubbles' 2d quantum geometry would consist with a $\SL(2,R)$ (irreducible unitary) representation and state attached to each surface patch and encoding its quantum state of 2d metric. Each graph edge would still carry a $\SU(2)$ representation -or spin- as in the standard formulation of loop quantum gravity. The balance equation along each edge, reflecting the gluing of bubbles, would equate the $\SU(2)$ Casimir on the edge with the $\SL(2,\R)$ Casimirs of the representations attached to its source and target,
\be
\mathfrak{C}_{\su_{2}}^{(e)}=\mathfrak{C}_{\sl_{2}}^{(e,s)}=\mathfrak{C}_{\sl_{2}}^{(e,t)}
\,,
\ee
thus implying a one-to-one correspondence between $\SU(2)$ representations (labeled by a half-integer spin $j\in\N/2$) and $\SL(2,\R)$ representations (from the discrete principal series of unitary representations with positive quadratic Casimir). This equilibrium between extrinsic and intrinsic curvatures of the surfaces would be interpreted as the quantum Gauss-Codazzi equation for bubbles. In this framework, a goal wold be to derive and study the algebra of symmetry and deformations of the quantum geometry of surfaces, and see if it can sustain a boundary conformal field theory in a continuum limit.

\item {\it The role of the surface graphs:}
The phase space and algebraic structure that we introduced for bubble network relies exclusively on the combinatorial structure of the 1-skeleton dual of the network of bubbles, as a graph whose nodes represent the bubbles and links indicate the gluing of two neighboring bubbles. However, considering that we describe the 2d metric on the bubbles' surfaces, it seems interesting to also keep track of the surface graph on each bubble, i.e. the 1-skeleton of network of surface patches for each bubble indicating which patches are neighbors as advocated in \cite{Feller:2017ejs}. This means keeping track of an extra layer of lower dimensional cells, as in the hierarchy of cellular decompositions e.g. used to formulate discrete topological field theories. 
This layer of information would be especially useful when considering geometrical observables, such as the 2d curvature, on the surface of the bubbles, which involve derivatives of the 2d metric.
The natural question is what algebraic data live on the surface graph of each bubble and how they are related to the extended twisted geometry variables? It amounts to adding some structure to the bubbles and dressing the 1d lines between surface patches. In the context of spinning geometries, the holonomy of a specific ``spinning'' connection live on these triangulation edges, from which the normal vectors $\vJ$'s were reconstructed \cite{Freidel:2013bfa,Charles:2016xzi}. It would be enlightening to understand if a similar idea can be generalized to the bubble networks, and explore what kind of defects can be respectively associated to the 1d surface edges and to the 2d surface patches.

\item {\it The generalization to a non-vanishing cosmological constant $\Lambda \ne 0$:}
This could be achieved by either introducing a $q$-deformation of the $T^*\SU(2)$ phase space (e.g. using the $\SL(2,\C)$ Poisson-Lie group structure of $q$-deformed loop quantum gravity introduced in \cite{Bonzom:2014wva,Bonzom:2014bua,Dupuis:2014fya}) or by extending the Casimir balance equation relating the $\sl_{2}$ and $\su_{2}$ charges and encoding the relation between the intrinsic and extrinsic geometries of the bubbles' surfaces, or by a suitable mixture of those two ingredients.

\item {\it Comparison with other proposed extensions of spin networks:}
We should compare the extended twisted geometry phase space proposed here with the Drinfeld tube networks based on the Drinfeld double $\cD\SU(2)$ and proposed in \cite{Dittrich:2016typ,Delcamp:2016yix,Dittrich:2017nmq} to account for both curvature and torsion defects, or the double spin network structures advocated in \cite{Charles:2016xzi} for studying the coarse-graining of loop quantum gravity. %
This would also shed light on the relation with the more complete picture proposed in \cite{Freidel:2016bxd} of loop gravity coupled to the full Kac-Moody algebra of surface metric deformation modes, which seems to lead to generalized spin network states with the $\SU(2)$ fluxes and holonomies coupled to conformal field theories on the bubbles' surfaces.


\item {\it Implement a dynamics of the bubble networks:}
We have introduced a kinematical framework for bubble networks, but the aim of quantum gravity is to define the dynamics of quantum geometry (and the action of diffeomorphisms -change of observers- at the quantum level). Towards this goal, it would be very interesting to try to implement the hydrodynamical formulation of the dynamics of general relativity, as proposed in \cite{Freidel:2014qya}, and derive evolution laws for the bubbles and their 2d boundary geometry.
Our goal is to reach a better understanding of the dynamics of gravitational edge modes living on space-time boundaries.

\end{itemize}

%

\appendix

\section{Holonomy reconstruction from the canonical pair of vectors}
\label{app-SO3holonomy}

\begin{lemma}
Let us consider a pair of 3-vectors $(\vx,\vp)$ such that $|\vx\w\vp|\ne 0$. We consider the symplectic generators:
$$
\ell_{0}=\vx\cdot\vp\,,\quad
\ell_{-}=|\vx|^{2}\,,\quad
\ell_{+}=|\vp|^{2}\,,\quad
\cD=(\ell_{-}\ell_{+}-\ell_{0}^{2})=|\vx\w\vp|^{2}\ne 0\,.
$$
There exists a unique rotation $h_{\vx,\vp}\in\SO(3)$ mapping the reference pair $(|\vx|\hat{e}_{1},\vv)$ to $(\vx,\vp)$ with:
$$
|\vx|\hat{e}_{1}
=
|\vx|
\mat{c}{1\\0\\0}
\,,\quad
\vv
=
\f1{|\vx|}
\mat{c}{\vx\cdot\vp \\ |\vx\w\vp|\\0 }
=
\f1{\sqrt{\ell_{-}}}
\mat{c}{\ell_{0} \\\sqrt{\cD}\\0}
\,,
$$
which is given by:
\be
h_{\vx,\vp}
=
\Bigg{(}
\f{\vx}{|\vx|}
\,,\,
\f{(\vx\w\vp)\w\vx}{|\vx|\,|\vx\w\vp|}
\,,\,
\f{\vx\w\vp}{|\vx\w\vp|}
\Bigg{)}
=
\Bigg{(}
\f{\vx}{|\vx|}
\,,\,
\f{|\vx|^{2}\vp-(\vx\cdot\vp)\vx}{|\vx|\,|\vx\w\vp|}
\,,\,
\f{\vx\w\vp}{|\vx\w\vp|}
\Bigg{)}
\ee
\end{lemma}
\begin{proof}
The matrix $h_{\vx,\vp}$ maps the $(Oxy)$ plane to the plane spanned by the two vectors $(\vx,\vp)$ and sends the direction $\hat{e}_{3}$ to the angular momentum $\vJ$. One simply needs to check that the three columns of $h$ form a positive orthonormal basis of $\R^{3}$ to prove that $h\in\SO(3)$.
\end{proof}

\noindent
Now we can combine the two rotations $h_{\vx^{s},\vp^{s}}$ and $h_{\vx^{t},\vp^{t}}$ to get the $\SO(3)$ holonomy living along the oriented edge $e$:
\be
h
=
h_{\vx^{t},\vp^{t}}\,\big{(}h_{\vx^{s},\vp^{s}}\big{)}^{-1}
=
h_{\vx^{t},\vp^{t}}\,{}^{t}h_{\vx^{s},\vp^{s}}\,.
\ee



\bibliographystyle{bib-style}
\bibliography{twisted-new}

\end{document}